\documentclass[submission,copyright,creativecommons]{eptcs}
\providecommand{\event}{AiML 2026} 

\usepackage{amsmath,amssymb,amsthm,amsfonts}
\usepackage{hyperref}
\usepackage{bussproofs}
\usepackage{mathrsfs}
\usepackage{enumerate}
\usepackage{stmaryrd}
\usepackage{tikz}
\usetikzlibrary{patterns,arrows}
\tikzset{shorten <>/.style={shorten >=#1, shorten <=#1}}
\usetikzlibrary{patterns,arrows}
\usetikzlibrary{automata, positioning}
\usepackage{setspace}

\newcommand{\ve}{\varnothing}
\newcommand{\powerset}[1]{{\mathcal{P}(#1)}}

\newcommand{\rsto}{\upharpoonright}
\newcommand{\sub}{\subseteq}
\newcommand{\set}[1]{{\{ #1 \}}}
\newcommand{\tup}[1]{{\langle #1 \rangle}}
\newcommand{\pair}[1]{{(#1)}}

\newcommand{\inver}[1]{{{#1}^{-1}}}

\newcommand{\ModelStyle}[1]{\mathfrak{#1}}
\newcommand{\mm}{\ModelStyle{M}}

\newcommand{\FrameStyle}[1]{\mathfrak{#1}}
\newcommand{\ff}{\FrameStyle{F}}

\newcommand{\GeneralFrameStyle}[1]{\mathbb{#1}}
\newcommand{\gf}{\GeneralFrameStyle{F}}

\newcommand{\ClassOfStructureStyle}[1]{\mathcal{#1}}
\newcommand{\CK}{\ClassOfStructureStyle{K}}

\newcommand{\FunctorStyle}[1]{\mathsf{#1}}

\newcommand{\Fin}{\FunctorStyle{Fin}}
\newcommand{\Fr}{\FunctorStyle{Fr}}

\newcommand{\GFr}{\FunctorStyle{GFr}}

\newcommand{\Log}{\FunctorStyle{Log}}

\newcommand{\NExt}{\FunctorStyle{NExt}}

\newcommand{\Prop}{\mathsf{Prop}}
\renewcommand{\phi}{\varphi}

\newsavebox{\necessity}
\sbox\necessity{%
\begin{tikzpicture}
\draw (0,0) rectangle (1.5ex,1.5ex);%
\end{tikzpicture}}
\newcommand{\B}{\operatorname{\usebox\necessity}}

\newsavebox{\possibility}
\sbox\possibility{%
\begin{tikzpicture}
  \draw[rotate=45] (0,0) rectangle(1.07ex,1.07ex);%
\end{tikzpicture}}
\newcommand{\D}{\operatorname{\usebox\possibility}}

\newsavebox{\BlackNecessity}
\sbox\BlackNecessity{%
\begin{tikzpicture}
\draw[fill=black] (0,0) rectangle(1.5ex,1.5ex);%
\end{tikzpicture}}
\newcommand{\bb}{\operatorname{\usebox\BlackNecessity}}

\newsavebox{\BlackPossibility}
\sbox\BlackPossibility{%
\begin{tikzpicture}
  \draw[rotate=45, fill=black] (0,0) rectangle(1.07ex,1.07ex);%
\end{tikzpicture}}
\newcommand{\bd}{\operatorname{\usebox\BlackPossibility}}

\newcommand{\SucPre}[1]{{#1}_\sharp}
\newcommand{\R}{\SucPre{R}}

\newcommand{\ML}[1]{\mathsf{#1}}
\newcommand{\TL}[1]{\mathsf{#1}_t}
\renewcommand{\L}{\mathscr{L}}

\newcommand{\Lt}{\mathscr{L}_t}
\newcommand{\md}{\models}

\newcommand{\M}{\mathsf{M}}

\newcommand{\K}{\mathsf{K}}
\usepackage{booktabs}
\usepackage{mathtools}

\renewcommand{\B}{\Box}
\renewcommand{\D}{\Diamond}
\renewcommand{\bd}{\blacklozenge}
\renewcommand{\bb}{\blacksquare}

\newcommand{\commentout}[1]{}

\newcommand{\is}[1]{{\Delta^{\leq #1}}}
\newcommand{\all}[1]{{\nabla^{\leq #1}}}
\usepackage{aiml26}

\usepackage{iftex}

\ifpdf
  \usepackage{underscore}         
  \usepackage[T1]{fontenc}        
\else
  \usepackage{breakurl}           
\fi

\title{Most Properties are Undecidable for Transitive Tense Logics}
\author{Qian Chen
\institute{
The Tsinghua-UvA JRC for Logic\\
Department of Philosophy\\
Tsinghua University\\
Beijing, China
}
\email{chenq21@mails.tsinghua.edu.cn}
\and
Tenyo Takahashi
\institute{Institute for Logic, Language and Computation\\
University of Amsterdam\\
Amsterdam, The Netherlands}
\email{t.takahashi@uva.nl}
}

\newcommand{\titlerunning}{Most Properties are Undecidable for Transitive Tense Logics}
\newcommand{\authorrunning}{Q. Chen \& T. Takahashi}

\hypersetup{
  bookmarksnumbered,
  pdftitle    = {\titlerunning},
  pdfauthor   = {\authorrunning},
  pdfsubject  = {EPTCS},               
  pdfkeywords = {modal logic, tense logic, decidability, logics' property, lattice of logics} 
}

\begin{document}
\maketitle

\begin{abstract}
    A logics' property is decidable in a class of logics if there exists an algorithm that decides whether a finitely axiomatizable logic in the class has the property. Many properties are undecidable for bimodal logics but decidable for linear tense logics, which leads to a general question on how the interactions of modalities affect the decidability of properties. In this paper, we study the decidability of properties for transitive tense logics and show that most properties are undecidable in the lattice $\NExt\TL{K4}$ of transitive tense logics, including Kripke completeness, the finite model property, and decidability. Our proof method adapts Chagrov’s approach of constructing a reduction from an undecidable problem of Minsky machines to the decision problem for logics' properties, yielding a general scheme of proving the undecidability of these properties.
\end{abstract}

\section{Introduction}

A central part of the study of modal logic is to determine whether a logic has a certain property, such as Kripke completeness, the finite model property (FMP), and decidability. From a global viewpoint, this gives rise to the following algorithmic problems: In a class of logics, is it decidable whether a logic has a certain property?\footnote{This should not be confused with the decidability of logics; a logic is \emph{decidable} if its membership problem is decidable.} More precisely, a logics' property (or simply, property) $P$ in a class $\mathcal{C}$ of logics is identified with a subclass of $\mathcal{C}$, and $P$ is \emph{decidable} if there exists an algorithm such that, for any finitely axiomatizable logic $L \in \mathcal{C}$, given by its finite axiomatization, the algorithm decides whether $L$ has the property $P$.\footnote{We have to restrict to finitely axiomatizable logics because an input of an algorithm must be a finite object; see Subsection~\ref{subsec2.2} and \cite[Section 17.1]{Chagrov.Zakharyaschev1997} for more discussions.} We refer to \cite[Section 17.6]{Chagrov.Zakharyaschev1997} and \cite{Wolter.Zakharyaschev2007} for historical accounts. See Table~\ref{table:UndecidableProperties} for a summary of the results discussed in this introduction.

Let $\K_n$ denote the least normal $n$-modal logic and $\K = \K_1$. For a normal modal logic $L$, let $\NExt L$ be the lattice of all normal extensions of $L$. In the unimodal case, many properties have been shown to be undecidable in $\NExt\K$. Thomason~\cite{Thomason1982} showed the undecidability of Kripke completeness, and a series of works by Chagrov and his co-authors \cite{chagrov1990I, chagrov1990II, Chagrov.Zakharyaschev1993, Chagrov.Chagrova1995, Chagrov2002} introduced a general method to show the undecidability of various properties, including the FMP, first-order definability, decidability, tabularity, and the coincidence with a fixed tabular logic. Kracht and Wolter~\cite{Kracht.Wolter1999} proved independently that decidability, the FMP, and tabularity are undecidable in $\NExt\mathsf{K}$ via the Thomason-Simulation \cite{Thomason1974a,Thomason1975b}. On the positive side, it is well-known that consistency is decidable for normal unimodal logics as the lattice $\NExt\K$ has only two coatoms \cite{Makinson1971} (see also \cite[Section 17.2]{Chagrov.Zakharyaschev1997}). Recently, Takahashi~\cite{Takahashi2026} proved that the property of being a union-splitting of $\NExt\mathsf{K}$ and strict Kripke completeness are decidable in $\NExt\K$ (see also \cite{TakahashiThesis}).

If $P$ is decidable in $\NExt L$, then it is also decidable in $\NExt L'$ for any extension $L'$ of $L$ with finitely many axioms (Proposition~\ref{prop:finite-extension}). On the other hand, a property that is undecidable in a given lattice $\mathcal{C}$ of logics may become decidable in a sublattice of $\mathcal{C}$. For example, in the lattice $\NExt\mathsf{K4}$ of transitive modal logics (recall that $\mathsf{K4} = \mathsf{K} \oplus \Box p \to \Box \Box p$), a sublattice of $\NExt\mathsf{K}$, every tabular logic has only finitely many immediate predecessors and all of them are tabular \cite{Blok1980}, so the coincidence with a fixed tabular logic is decidable in this setting \cite{Jankov1968b,Rautenberg1979a} (see also \cite[Theorem 17.3]{Chagrov.Zakharyaschev1997}). Moreover, as far as we know, the decidability of tabularity in $\NExt\ML{K4}$ is still open \cite{rautenbergWillemBlokModal2006}, whereas it is undecidable in $\NExt\mathsf{K}$ as mentioned above.

In this paper, we study the decision problem for properties for \emph{tense logics}. Tense logics are normal bimodal logics extending the logic $\TL{K} = \K_2 \oplus \set{p \to \B\bd p, p \to \bb\D p}$, where, following convention, we denote the two modalities $\B$ and $\bb$ with their duals $\D$ and $\bd$. The intended meaning of $\B$ and $\bb$ are ``always true in the future'' and ``always true in the past'', respectively. The least tense logic $\TL{K}$ was introduced in the 1960s by Prior \cite{Prior1967,Prior1968}. Philosophically, tense logics can be viewed as logics of time; for example, $\TL{K4} = \TL{K} \oplus \D\D p \to \D p$ is the logic of transitive time flows, and $\TL{Lin} = \TL{K4} \oplus (\D\bd p \vee \bd\D p \to p \vee \D p \vee \bd p)$ is the logic of linear time flows (see also \cite{Wolter.Zakharyaschev2007}). Algebraically, the modalities $\B$ and $\bd$ are adjoint, in the sense that $\bd\phi \to \psi \in L$ if and only if $\phi \to \B\psi \in L$ for any tense logic $L$. Thus, tense logics may also be viewed as ``logics of adjointness,'' and they serve as examples of bimodal logics in which two modalities interact in a natural way.

The bimodal logics $\ML{K}_2 \sub \TL{K} \sub \TL{K4} \sub \TL{Lin}$ form a chain in $\NExt\K_2$. Many properties, including Kripke completeness, the FMP, decidability, and even consistency, are undecidable in $\NExt\K_n$ for $n \geq 2$ \cite{Thomason1982}. On the other hand, the aforementioned properties are all decidable in $\NExt\TL{Lin}$ \cite{Wolter1996a,Wolter1997}. These results indicate that the interactions of modalities significantly affect the decidability of properties, which naturally raises the question of whether these properties are decidable in $\NExt\TL{K}$ or $\NExt\TL{K4}$. Chagrov and Shehtman~\cite{Chagrov.Shehtman1995} proved that tabularity, the coincidence with a fixed tabular tense logic, and consistency are undecidable in $\NExt\TL{K4}$ (and thus also undecidable in $\NExt\TL{K}$), while the decidability of other properties, such as Kripke completeness, the FMP, and decidability, remained open.

Although $\TL{K}$ and $\TL{K4}$ are the minimal tense extensions of $\ML{K}$ and $\ML{K4}$ respectively, the undecidability of properties in $\NExt\TL{K}$ and $\NExt\TL{K4}$ does not follow directly from the undecidability results for $\NExt\mathsf{K4}$ and $\NExt\mathsf{K}$. The minimal tense extension map $(\cdot)_t: \NExt\mathsf{K} \to \NExt\mathsf{K}_t$ is not injective~\cite{Wolter1993}, while it remains unknown whether $(\cdot)_t{\rsto}\NExt\mathsf{K4}$ is injective \cite[p. 158]{Wolter1997a}. Moreover, Wolter~\cite{Wolter1996} presented a modal logic $L \in \NExt\mathsf{K4}$ having the FMP whose minimal tense extension $L_t \in \NExt\TL{K4}$ is Kripke incomplete. It follows that the map $(\cdot)_t$ preserves neither Kripke completeness nor the FMP, while it remains open whether decidability is preserved (see \cite[p. 132]{Wolter1997a}). The interactions between tense modalities make lattices of tense logics completely different from those of unimodal logics (see \cite{Kracht1992,Ma.Chen2021,Ma.Chen2023,Chen.Ma2024}).

In this paper, we provide a general criterion (Theorem~\ref{thm:reduction-K4t-P}) for a property to be undecidable in $\NExt\TL{K4}$. Most properties studied for logics fall into this criterion, including Kripke completeness, the FMP, and decidability; see Corollary~\ref{Thm:undec-properties} for a more comprehensive list of undecidable properties in $\NExt\TL{K4}$ following from the theorem. It follows that these properties are also undecidable in $\NExt\TL{K}$.

Our proof adapts the method of \cite{Chagrov.Shehtman1995}, reducing an undecidable problem regarding Minsky machines to the decision problem for a property. \emph{Minsky machines}, also called counter machines or register machines, are a type of mathematical model of computation as strong as Turing machines \cite{minskyComputationFiniteInfinite1967}. There exist a Minsky machine $\M$ and a configuration $c_0$ of $\M$ such that it is undecidable whether a given configuration of $\M$ is reachable from $c_0$ by computation of $\M$ (see, e.g., \cite[Theorem 16.3]{Chagrov.Zakharyaschev1997}). Let us call this undecidable problem $Q$. Given a property $P$, we will find a logic $L \in \NExt\TL{K4}$ that has $P$ and construct a computable reduction from $Q$ to the decision problem for $P$ as follows. For a configuration $c$, the reduction produces a logic $L(c) \in \NExt\TL{K4}$ such that: 
\begin{enumerate}
    \item if $c$ is reachable from $c_0$, then $L(c) = L$, which implies that $L(c)$ has $P$;
    \item if $c$ is not reachable from $c_0$, then $L(c)$ does not have $P$.
\end{enumerate}
Thus, if we could decide whether a logic in $\NExt\TL{K4}$ has $P$, we would be able to decide the problem $Q$: Given a configuration $c$, we compute the logic $L(c)$ and ask if it has the property $P$, the answer of which is also the answer to $Q$. Since $Q$ is undecidable, it follows that $P$ is undecidable.

Table~\ref{table:UndecidableProperties} summarizes the results on the decidability of some major properties discussed so far. The entries marked with \textsuperscript{\dag} are results established in this paper.

\begin{table}[htbp]
    \centering
    {\renewcommand{\arraystretch}{1.5}
    \begin{tabular}{c|c|c|c|c|c|c}
        & Cons. & Tab. & Fixed Tab. & KC & FMP & Dec. \\ \hline
        $\NExt\ML{K}$ 
        & $\checkmark$ & $\times$ & $\times$ & $\times$ & $\times$ & $\times$ \\ \hline

        $\NExt\ML{K4}$ 
        & $\checkmark$ & ? & $\checkmark$ & $\times$ & $\times$ & $\times$ \\ \hline

        $\NExt\ML{K_2}$ 
        & $\times$ & $\times$ & $\times$ & $\times$ & $\times$ & $\times$ \\
        \hline 

        $\NExt\TL{K}$ 
        & $\times$ & $\times$ & $\times$ 
        & $\times$\textsuperscript{\dag} 
        & $\times$\textsuperscript{\dag} 
        & $\times$\textsuperscript{\dag} \\ \hline

        $\NExt\TL{K4}$  
        & $\times$ & $\times$ & $\times$ 
        & $\times$\textsuperscript{\dag} 
        & $\times$\textsuperscript{\dag} 
        & $\times$\textsuperscript{\dag} \\ \hline

        $\NExt\TL{Lin}$ 
        & $\checkmark$ & $?$ & $\checkmark$ & $\checkmark$ & $\checkmark$ & $\checkmark$
    \end{tabular} 
    }

    \vspace{0.5em}
    {\footnotesize
    \textbf{Abbreviations.}
    Cons.: consistency;
    Tab.: tabularity;
    Fixed Tab.: coincidence with a fixed tabular logic;\\
    KC: Kripke completeness;
    FMP: finite model property;
    Dec.: Decidability.
    }
    \caption{(Un)decidability of logics' properties in lattices of unimodal and bimodal logics}
    \label{table:UndecidableProperties}
\end{table}

This paper is organized as follows. Section~\ref{Sec 2} introduces preliminaries on tense logics and Minsky machines. In section~\ref{Sec 3}, we prove the main theorem and apply it to show the undecidability of logics' properties. Finally, Section~\ref{Sec 4} concludes the paper with an overview of future work.

\section{Preliminaries} \label{Sec 2}

\subsection{Preliminaries of modal and tense logic}

Recall that for each $n \in \mathbb{Z}^+$, the \textit{$n$-modal language} $\L_n$ is obtained by adding the modalities $\B_0, \cdots, \B_{n-1}$ to the propositional language. A \textit{normal $n$-modal logic} is a set of $\L_n$-formulas which contains all classical tautologies, K-axioms $\B_i(p \to q) \to (\B_i p \to \B_i q) \in L$ for all $i < n$, and is closed under the rules Modus Ponens ($\phi,\phi \to \psi/\psi$), Substitution ($\phi(p_1,\cdots,p_n)/\phi(\psi_1,\cdots,\psi_n)$) and Necessitation ($\phi/\B_i\phi$). Let $\ML{K}_n$ denote the minimal normal $n$-modal logic. We refer to $2$-modal logics as \textit{bimodal logics}.

A \textit{tense logic} is a normal bimodal logic containing the axioms $p \to \B_0\D_1 p$ and $p \to \B_1\D_0 p$. For tense logics, we write $\Lt$ for the formal language, $\B$ for $\B_0$ and $\bb$ for $\B_1$. Let $\TL{K}$ denote the \textit{minimal tense logic}. For each tense logic $L$, let $\NExt L$ denote the lattice of all normal extensions of $L$. For every tense logic $L$ and set of formulas $\Sigma$, let $L\oplus\Sigma$ denote the smallest tense logic containing $L\cup\Sigma$. We write $L \oplus \phi$ for $L \oplus \set{\phi}$. A tense logic $L$ is \emph{finitely axiomatizable} if $L = \TL{K} \oplus \phi$ for some $\phi \in \Lt$.
A tense logic $L$ is \emph{consistent} if $\bot\not\in L$, and so the only inconsistent tense logic is $\Lt$.

A \emph{general frame} is a triple $\gf=(X,R,A)$ where $X$ is a non-empty set, $R$ a binary relation on $X$ and $A$ a subset of $\powerset{X}$ such that (i) $\ve \in A$, and (ii) $A$ is closed under Boolean operations on $\powerset{X}$, $R[\cdot]$ and $\inver{R}[\cdot]$, where $R[Y] \coloneqq \set{x \in X: \exists{y\in Y}(Ryx)}$ and $\inver{R}[Y] \coloneqq \set{x \in X: \exists{y\in Y}(Rxy)}$ for all $Y \sub X$. A \textit{Kripke frame} $\ff$ is a general frame of the form $(X,R,\powerset{X})$ and we simply write $(X,R)$. Let $\GFr$, $\Fr$, and $\Fin$ denote the classes of all general frames, Kripke frames, and finite Kripke frames, respectively.

A \emph{model} is a pair $\mm=(\gf,V)$ where $\gf\in\GFr$ and $V: \Prop \to A$ a valuation in $\gf$. $V$ is extended to $V:\Lt\to A$ as usual: $V(\bd\phi) = R[V(\phi)]$ and $V(\B\phi) \coloneqq X \setminus \inver{R}[X\setminus V(\phi)]$. The expressions $\mm,x \md \phi$, $\gf,x \md \phi$, $\gf \md \phi$ and $\gf \md \Sigma$ are defined as usual. Note that $\mm,x \md \bd \phi$ if and only if $\mm,y \md \phi$ for some $y \in \inver{R}[x]$. For all sets $\Sigma\sub\Lt$ of formulas and classes $\mathcal{K}\sub\GFr$ of general frames, let 
    \begin{center}
        $\mathcal{K}(\Sigma) \coloneqq \set{\gf\in\mathcal{K}:\gf\vDash\Sigma}$ and $\mathsf{Log}(\mathcal{K}) \coloneqq \set{\phi:\mathcal{K}\vDash\phi}$.
    \end{center}
For example, given a tense logic $L$, we write $\Fin(L)$ for the class of all finite frames validating $L$. We call $\mathsf{Log}(\mathcal{K})$ the \emph{tense logic of $\mathcal{K}$}.
Recall that a Kripke frame $(X,R)$ is called \textit{transitive} if $R$ is transitive. 
Then $\TL{K4} \coloneqq \TL{K} \oplus \D\D p \to \D p$ is the tense logic of transitive frames. In other words, we have $\TL{K4} = \mathsf{Log}(\set{\ff \in \Fr: \ff \text{ is transitive}})$. A tense logic $L$ is transitive if $L$ extends $\TL{K4}$, i.e., $L \supseteq \TL{K4}$.

Let us recall some properties of tense logics. Let $L$ be a tense logic. Then (i) $L$ is \textit{Kripke complete}, if $L = \Log(\Fr(L))$; (ii) $L$ has the \textit{finite model property} (\textit{FMP}), if $L = \Log(\Fin(L))$; (iii) $L$ is \textit{tabular}, if $L = \Log(\ff)$ for some finite frame $\ff$; (iv) $L$ is \emph{canonical}, if $L=\mathsf{Log}(\ff^L)$, where $\ff^L$ is the \textit{canonical frame} for $L$; (v) $L$ is first-order definable, if $L = \Log(\CK)$, where $\CK$ is a class of Kripke frames defined by a set of first-order sentences. Moreover, we say that $L$ is \emph{locally tabular} if for each $n \in \omega$, $L$ contains only finitely many non-$L$-equivalent formulas built up from the propositional variables $p_{0}, \cdots, p_{n-1}$. Finally, we say that $L$ is \emph{decidable} if there is an algorithm that, given a formula $\phi$, decides whether $\phi \in L$. This \emph{membership decidability} is a property for logics, and should not be confused with decidability of logics' properties; in particular, it is legitimate to say decidability, as a property, is decidable or not.

\subsection{Decision problem for logics' properties} \label{subsec2.2}

In this section, instead of the decidability of logics, we focus on the decidability of logics' properties.
We refer to \cite[Chapter 17]{Chagrov.Zakharyaschev1997} and \cite{Wolter.Zakharyaschev2007} for a detailed introduction and survey of the decision problem for properties of modal logics. We will work in the tense (or more generally, bimodal) setting. We identify a property $P$ in the lattice $\NExt(L_0)$ with the set of logics in $\NExt(L_0)$ that satisfy $P$, that is, $P = \{L \in \NExt(L_0): L \text{ satisfies } P\}$. 

\begin{definition} \label{def:decidable}
    Let $L_0$ be a tense logic. A property $P$ is \emph{decidable} in $\NExt(L_0)$ iff the set $\{\phi: L_0 \oplus \phi \in P\}$ is decidable.
\end{definition}

Following convention, we restrict ourselves to finitely axiomatizable logics because an input for an algorithm must be a finite object. We do not consider all recursively axiomatizable logics, as Kuznetsov showed that otherwise the only decidable properties would be the trivial ones (see \cite[Section 17.1]{Chagrov.Zakharyaschev1997}). Since most logics we encounter in practice are finitely axiomatizable, this is not a serious drawback. A finitely axiomatizable logic will be encoded by a finite set of formulas axiomatizing the logic, or equivalently, a single formula axiomatizing the logic.

Note that determining a property in a larger lattice of logics is at least as hard as in a smaller one, in the following sense.

\begin{proposition} \label{prop:finite-extension}
    Let $L \in \NExt\K_n$ and $L'$ be an extension of $L$ with finitely many axioms. If a property $P$ is undecidable in $\NExt L'$, then it is undecidable in $\NExt L$.
\end{proposition}

\begin{proof}
    We may assume $L' = L \oplus \phi$ for a formula $\phi$. We prove the contrapositive. Let $P$ be a decidable property in $\NExt L$. Then, given a formula $\psi$, we can determine whether $L' \oplus \psi$ has $P$ by asking whether $L \oplus \phi \land \psi$ has $P$ since the two logics are the same.
\end{proof}

\subsubsection*{Minsky machines}

The most commonly used method for proving the undecidability of a decision problem is to construct a computable reduction from another problem that is already known to be undecidable to the problem. In this paper, we will use an undecidable problem about Minsky machines. We recall the basics of Minsky machines in the rest of this section and refer to \cite[Section 16.1]{Chagrov.Zakharyaschev1997} and \cite{minskyComputationFiniteInfinite1967} for more details; see also \cite[Sections 16 and 17]{Chagrov.Zakharyaschev1997} for various applications of Minsky machines to obtain undecidability results.

A \emph{Minsky machine} with two registers (also called a counter/register machine with two counters/registers) is a finite set of instructions acting on two registers. We will only use Minsky machines with two registers, so we simply call them Minsky machines. A Minsky machine has finitely many \emph{states}. A \emph{register} can store a natural number and is assumed to be unbounded. So, a situation of a Minsky machine is represented by a tuple $\tup{s, n, m}$, called a \emph{configuration}, where $s$ is the current state and $n$ and $m$ are the natural numbers on each register. An instruction operates on the state and one of the two registers: it increments the number in the register, or tests if the number in the register is zero and decrements it if not. More specifically, an instruction $I$ has one of the following four forms: 

\begin{itemize}
    \item $I = t \to \tup{t', 1, 0}$ means that $I$ turns the state $t$ into $t'$ and increment the first register,
    \item $I = t \to \tup{t', 0, 1}$ means that $I$ turns the state $t$ into $t'$ and increment the second register,
    \item $I = t \to \tup{t', -1, 0} (\tup{t'', 0, 0})$ means that $I$ turns the state $t$ into $t'$ and decrements the first register if the number in the first register is non-zero, and turns the state $t$ into $t''$ otherwise,
    \item $I = t \to \tup{t',0, -1} (\tup{t'', 0, 0})$ means that $I$ turns the state $t$ into $t'$ and decrements the second register if the number in the second register is non-zero, and turns the state $t$ into $t''$ otherwise.
\end{itemize}

For example, applying the instruction $I = s \to \tup{s', -1, 0} (\tup{s'', 0, 0})$ to the configuration $\tup{s, n, m}$, we obtain the configuration $\tup{s', n-1, m}$ if $n \geq 1$ and the configuration $\tup{s'', n, m}$ if $n = 0$. 

In this paper, Minsky machines are assumed to be \emph{deterministic}, that is, for each state $t$ there is at most one instruction that acts on the state $t$. For a Minsky machine $\M$, we write $\M: \tup{s,n,m} \rightsquigarrow \tup{t,k,l}$ if, starting from the configuration $\tup{s,n,m}$, by applying the instructions in $\M$, we can reach the configuration $\tup{t,k,l}$ in finitely many (possibly 0) steps. We drop $\M$ if it is clear from the context. 

\begin{example}
    Let $\M = \set{s \to \tup{s, 1, 0}}$. Then, for any $n, m \in \omega$, 
    \[\set{\tup{t,k,l}: \tup{s,n,m} \rightsquigarrow \tup{t,k,l}} = \set{\tup{s, n+i, m}: i \in \omega}.\]
\end{example}

We will use the following undecidable problem, which is called the \emph{second configuration problem} in \cite[Theorem 16.3]{Chagrov.Zakharyaschev1997}. 

\begin{theorem} \label{Thm:undec-minsky}
    There exist a Minsky machine $\M$ and a configuration $\tup{s,n,m}$ such that the reachability from $\tup{s,n,m}$ in $\M$ is undecidable, that is, the set $\set{\tup{t,k,l}: \tup{s,n,m} \rightsquigarrow \tup{t,k,l}}$ is undecidable.
\end{theorem}

\section{Undecidable properties in $\NExt\TL{K4}$} \label{Sec 3}

The aim of this section is to prove the general undecidability result Theorem~\ref{thm:reduction-K4t-P}, which implies the undecidability of various properties summarized in Corollary~\ref{Thm:undec-properties}. The proof idea follows Chagrov's method of using Minsky machines \cite{chagrov1990I, chagrov1990II,Chagrov.Shehtman1995} (see also \cite{Wolter.Zakharyaschev2007}). Let $\M$ be a Minsky machine and $\tup{s,n,m}$ be a configuration of $\M$ such that the set $\set{\tup{t,k,l}: \M: \tup{s,n,m}\rightsquigarrow\tup{t,k,l}}$ is undecidable, given by Theorem~\ref{Thm:undec-minsky}. Since the Minsky machine $\M$ is finite, we may assume that $\M$ contains $t_0$ many states, labeled as $0, \dots, t_0-1$. To state our main theorem, we introduce the following general frame that encodes the problem $\set{\tup{t,k,l}:\tup{s,n,m}\rightsquigarrow\tup{t,k,l}}$.

\begin{definition}
    Let $\gf=(W,R,A)$ be the general frame defined as follows:
\begin{itemize}
    \item $W=\set{\tup{t,k,l}:\tup{s,n,m}\rightsquigarrow\tup{t,k,l}}\cup\set{a_n, b_n, c_n :n<\omega}\cup\set{a', b', b''}$;
    \item $R$ is the transitive closure of the union of the following binary relations:
    \begin{itemize}
        \item $\set{\pair{c_i,c_j}:j<i<\omega}$;
        \item $\set{\pair{a_i,a_j}:j<i<\omega}\cup\set{\pair{a',a_0}}$;
        \item $\set{\pair{b_i,b_j}:j<i<\omega}\cup\set{\pair{b',b_0},\pair{b',b''}}$;
        \item $\set{\pair{\tup{t,k,l},c_t},\pair{\tup{t,k,l},a_k},\pair{\tup{t,k,l},b_l} ,\pair{\tup{t,k,l},\tup{t,k,l}}:\tup{s,n,m}\rightsquigarrow\tup{t,k,l}}$.
    \end{itemize}
    \item $A = \set{U \sub W : \text{ $U$ is finite or cofinite on $\set{c_i : i \in \omega}$}}$.
\end{itemize}
\end{definition}

It is clear that $A$ is closed under $\cap$, $W\setminus(\cdot)$, $R[\cdot]$ and $\inver{R}[\cdot]$, so $\gf$ is well-defined as a general frame. The underlying Kripke frame $(W,R)$ of $\gf$ is depicted in Figure~\ref{fig:fi-K4t}. 
\begin{figure}[htbp]
    \small
\centering
\begin{tikzpicture}[scale=0.7]

\def\ptRad{.2pt}
\node (a2) at (0,6)[label=left:$c_0$]{$\bullet$};
\node (a3) at (0,5)[label=left:$c_1$]{$\bullet$};
\node (a4) at (0,4)[label=left:$c_2$]{$\bullet$};   
\node (ai0) at (0,2.5)[label=left:$c_{t-1}$]{$\bullet$};
\node (ai1) at (0,1.5)[label=left:$c_{t}$]{$\bullet$};
\node (ai2) at (0,0.5)[label=left:$c_{t+1}$]{$\bullet$};

\node (b2) at (2.50,6)[label=left:$a_0$]{$\bullet$};
\node (b3) at (2.50,5)[label=left:$a_1$]{$\bullet$};
\node (b4) at (2.50,4)[label=left:$a_2$]{$\bullet$};
\node (bi0) at (2.50,2.5)[label=left:$a_{k-1}$]{$\bullet$};
\node (bi1) at (2.50,1.5)[label=left:$a_{k}$]{$\bullet$};
\node (bi2) at (2.50,0.5)[label=left:$a_{k+1}$]{$\bullet$};

\node (c2) at (5.0,6)[label=left:$b_0$]{$\bullet$};
\node (c3) at (5.0,5)[label=left:$b_1$]{$\bullet$};
\node (c4) at (5.0,4)[label=left:$b_2$]{$\bullet$};
\node (ci0) at (5.0,2.5)[label=left:$b_{l-1}$]{$\bullet$};
\node (ci1) at (5.0,1.5)[label=left:$b_{l}$]{$\bullet$};
\node (ci2) at (5.0,0.5)[label=left:$b_{l+1}$]{$\bullet$};

\node (vd1) at (5.0,3.4){$\vdots$};
\node (vd2) at (0,3.4){$\vdots$};
\node (vd3) at (2.50,3.4){$\vdots$};

\node (vd4) at (0,0){$\vdots$};
\node (vd5) at (5.0,0){$\vdots$};
\node (vd6) at (2.50,0){$\vdots$};

\node (z2) at (5.0+2,6)[label=below:$b''$]{$\bullet$};
\node (z1) at (5.0+1,5)[label=below:$b'$]{$\bullet$};
\node (y1) at (2.50+1,5)[label=below:$a'$]{$\bullet$};

\node (s) at (2.0,-1){$\bullet$};
\draw[loop below, ->] (s) to (s);
\node (s) at (2.0,-1)[label=left:$\cdots$]{};
\node (s) at (2.0,-1)[label=right:$\tup{t,k,l}$]{};

\draw [->,shorten <>=\ptRad] (a3) -- (a2);
\draw [->,shorten <>=\ptRad] (a4) -- (a3);
\draw [->,shorten <>=\ptRad] (ai1) -- (ai0);
\draw [->,shorten <>=\ptRad] (ai2) -- (ai1);

\draw [->,shorten <>=\ptRad] (b3) -- (b2);
\draw [->,shorten <>=\ptRad] (b4) -- (b3);
\draw [->,shorten <>=\ptRad] (bi1) -- (bi0);
\draw [->,shorten <>=\ptRad] (bi2) -- (bi1);

\draw [->,shorten <>=\ptRad] (c3) -- (c2);
\draw [->,shorten <>=\ptRad] (c4) -- (c3);
\draw [->,shorten <>=\ptRad] (ci1) -- (ci0);
\draw [->,shorten <>=\ptRad] (ci2) -- (ci1);

\draw [->,shorten <>=\ptRad] (z1) -- (z2);
\draw [->,shorten <>=\ptRad] (z1) -- (c2);
\draw [->,shorten <>=\ptRad] (y1) -- (b2);

\draw [->,shorten <>=\ptRad] (s) -- (ai1);
\draw [->,shorten <>=\ptRad] (s) -- (bi1);
\draw [->,shorten <>=\ptRad] (s) -- (ci1);
\end{tikzpicture}

\caption{The underlying frame of $\gf$}
\label{fig:fi-K4t}
\end{figure}

Before studying properties of the general frame $\gf$, let us introduce a set of new modal operators $\is{n}$ and their duals $\all{n}$, which will play an important role in our proofs.

\begin{definition}
    For each $n\in\omega$ and $\phi,\psi\in\Lt$, we define the formula $\is{n}\phi$ by:
\begin{center}
    $\is{0}\phi=\phi$ and $\is{k+1}\phi=\is{k}\phi\vee\D\is{k}\phi\vee\bd\is{k}\phi$.
\end{center}
As usual, we define the dual operator $\all{n}$ of $\is{n}$ by $\all{n}\phi\coloneqq\neg\is{n}\neg\phi$. 
\end{definition}
For each Kripke frame $\ff = (X,R)$, let $\R$ be the binary relation $(= \cup R \cup \inver{R})$ on $X$, that is,
\begin{center}
    $\R = \set{(x,y) \in X \times X : x=y \text{ or } Rxy \text{ or } Ryx}$.
\end{center}
It follows that for all $x,y \in X$, if $y \in \R^{n}[x]$, then there exists a sequence $C = \tup{x_i : i \leq k}$ with $k < n$ such that (i) $x = x_0$, (ii) $y = x_k$, and (iii) $x_{i+1} \in \R[x_{i}]$ for all $i < k$. We call $C$ a \emph{$k$-path} from $x$ to $y$.

\begin{proposition}
    Let $\mm=(X,R,V)$ be a model, $x\in X$ and $\phi\in\Lt$. Then for all $k\in\omega$,
    \begin{center}
        $\mm,x\md\is{k}\phi$ if and only if $\mm,y\md\phi$ for some $y\in\R^k[x]$.
    \end{center}
\end{proposition}
\begin{proof}
    By induction on $k$.
\end{proof}

\begin{lemma}\label{lem:FMsnm-basic} 
    $\gf\md\is{6}p\to\is{5}p$.
\end{lemma}
\begin{proof}
    Any two points in $W$ are connected by an $(R\cup R^{-1})$-chain of length no more than $5$.
\end{proof}

Now we state the main theorem. Let 
\[\mathsf{KC} = \set{L \in \NExt\TL{K4} : L \text{ is Kripke complete}}\]
and
\[\mathsf{DEC} = \set{L \in \NExt\TL{K4} : L \text{ is decidable}}.\]

\begin{theorem}\label{thm:reduction-K4t-P}
    Let $P$ be a property and $\alpha \in \Lt$ be a formula such that $\gf\not\md\alpha$ and $\TL{K4}\oplus\alpha \in P$ and $P \sub \mathsf{KC} \cup \mathsf{DEC}$. Then $P$ is undecidable, i.e., the set $\set{\phi\in\Lt:\TL{K4}\oplus\phi \in P }$ is undecidable.
\end{theorem}

The rest of this section is dedicated to proving this theorem. We start by introducing formulas that define a point or a set of points in $\gf$. Their meaning is summarized in Lemma~\ref{lem:FMsnm-ch}. For each $n\in\omega$, we define the formulas $\phi_{c_{n}}$, $\phi_{a_{n}}$ and $\phi_{b_{n}}$ as follows:
\begin{itemize}
    \item $\phi_{c_{0}}\coloneqq\B\bot\land\bb\bd\top$ and $\phi_{c_{k}}\coloneqq\D^k\phi_{c_{0}}\wedge\neg\D^{k+1}\phi_{c_{0}}$;
    \item $\phi_{a_{0}}\coloneqq\B\bot\wedge\bd(\bb\bot\wedge\B\bd^2\top)$ and $\phi_{a_{k}}\coloneqq(\phi_{a_{0}} \lor \D \phi_{a_{0}}) \land \B\neg\phi_{c_{0}}\land\bd\top \land \D^k\phi_{a_{0}}\wedge\neg\D^{k+1}\phi_{a_{0}}$;
    \item $\phi_{b_{0}}\coloneqq\B\bot\wedge\bd\bd\top\land\bd\D\bb^2\bot$ and $\phi_{b_{k}}\coloneqq(\phi_{b_{0}} \lor \D \phi_{b_{0}}) \land \B\neg\phi_{c_{0}} \land\bd\top \land \D^k\phi_{b_{0}}\wedge\neg\D^{k+1}\phi_{b_{0}}$. 
\end{itemize}
Note that these formulas are all variable-free. Intuitively, for each $x \in \set{a,b,c}$ and $n \in \omega$, the formula $\phi_{x_n}$ is designed to be true at exactly $x_n$ in $\gf$, regardless of valuations. Moreover, we define 
\begin{itemize}
    \item $\phi_{A} \coloneqq (\phi_{a_{0}} \lor \D \phi_{a_{0}}) \land \B\neg\phi_{c_{0}} \land \bd\top$;
    \item $\phi_{B} \coloneqq (\phi_{b_{0}} \lor \D \phi_{b_{0}}) \land \B\neg\phi_{c_{0}}\land\bd\top$.
\end{itemize}
Similarly, $\phi_{A}$ and $\phi_{B}$ are true exactly at $\set{a_n: n \in \omega}$ and $\set{b_n: n \in \omega}$ in $\gf$, respectively. To simulate the action of $+1$ and $-1$ on the two tapes of $\M$, we define the following formulas:
\begin{itemize}
    \item $\psi_{A} \coloneqq \phi_{A} \land p_A \land \lnot \D p_A$;
    \item $\psi^+_{A} \coloneqq \phi_{A} \land \D p_A \land \lnot\D\D p_A$;
    \item $\psi_B \coloneqq \phi_{B} \land p_B \land \lnot\D p_B$;
    \item $\psi^+_B \coloneqq \phi_{B} \land \D p_B \land \lnot\D\D p_B$,
\end{itemize}
where $p_A,p_B$ are fresh variables. Intuitively, if $\psi_A$ is true at some point, then the point must be $a_i$, where $i = \min\{j: a_j \models p_A\}$; then $\psi^+_A$ is true at the next point, namely, $a_{i+1}$. A similar intuition applies to $\psi_B$ and $\psi^+_B$ as well.
Moreover, a key syntactic observation is: if $s$ is the substitution $[\D^k \phi_{a_{0}} / p_A, \D^l \phi_{b_{0}} / p_B]$ for some $k, l \in \omega$, then $s(\psi_{A}) = \phi_{a_{k}}$, $s(\psi^+_{A}) = \phi_{a_{k+1}}$, $s(\psi_B) = \phi_{b_{l}}$, and $s(\psi^+_B) = \phi_{b_{l+1}}$. This will be used in Lemma~\ref{lem:AxM-works}.

Finally, for each state $t$ of $\M$ and formulas $\pi,\kappa\in\Lt$, we define:
\[\sigma(t,\pi,\kappa) \coloneqq \D\phi_{c_{t}} \wedge \B\neg\phi_{c_{t+1}} \wedge \D\pi \wedge \B\lnot(\D\pi \land \B\lnot\phi_{c_{0}}) \wedge \D\kappa \wedge \B\lnot(\D\kappa \land \B\lnot\phi_{c_{0}}).\]
As we will see in Lemma~\ref{lem:FMsnm-ch}, the formula $\sigma(t,\pi,\kappa)$ is true exactly at the point $\tup{t,k,l}$ if the formulas $\pi$ and $\kappa$ are true exactly at $a_k$ and $b_l$, respectively.

\begin{lemma}\label{lem:FMsnm-ch}
    For any $w\in W$, valuation $V$ on $\gf$, and $n\in\omega$, the following holds:
    \begin{enumerate}[(1)]
        \item $\gf,w\md\phi_{c_{n}}$ if and only if $w=c_n$.
        \item $\gf,w\md\phi_{a_{n}}$ if and only if $w=a_n$.
        \item $\gf,w\md\phi_{b_{n}}$ if and only if $w=b_n$.
        \item $\gf,w\md\phi_{A}$ if and only if $w \in \set{a_i : i \in \omega}$.
        \item $\gf,w\md\phi_{B}$ if and only if $w \in \set{b_i : i \in \omega}$.
        \item $\gf,V,w\md\psi_{A}$ if and only if $V(\psi_{A}) = \set{a_i} = \set{w}$ and $V(\psi^+_{A}) = \set{a_{i+1}}$ for some $i \in \omega$.
        \item $\gf,V,w\md\psi^+_{A}$ if and only if $V(\psi_{A}) = \set{a_i}$ and $V(\psi^+_{A}) = \set{a_{i+1}} = \set{w}$ for some $i \in \omega$.
        \item $\gf,V,w\md\psi_B$ if and only if $V(\psi_B) = \set{b_i} = \set{w}$ and $V(\psi^+_B) = \set{b_{i+1}}$ for some $i \in \omega$.
        \item $\gf,V,w\md\psi^+_B$ if and only if $V(\psi_B) = \set{b_i}$ and $V(\psi^+_B) = \set{b_{i+1}} = \set{w}$ for some $i \in \omega$.
        \item If $V(\pi) = \{a_k\}$ and $V(\kappa) = \{b_l\}$, then $\gf, V, w \md \sigma(t, \pi, \kappa)$ if and only if $w = \tup{t,k,l}$.
    \end{enumerate}
\end{lemma}
\begin{proof}
    We only prove (1), (4), (6), and (10) and leave the rest to the readers. The proof of (1) proceeds by induction on $n$. Let $n=0$. The right-to-left direction is clear. Suppose $\gf,w\md\phi_{c_{0}}$. Then $R[w]=\ve$ and $R^{-1}[u]\neq\ve$ for all $u\in R^{-1}[w]$, which entails $w=c_0$. Let $n>0$. Suppose $\gf,w\md\phi_{c_{n}}$. By the induction hypothesis, $c_{n-1}\in R[w]\setminus R[R[w]]$, thus $w=c_n$, and (1) follows. 

    For (4), the right-to-left direction is straightforward. For the other direction, suppose $\gf,w \md \phi_{A}$. Then $\gf,w \md \phi_{a_{0}} \lor \D \phi_{a_{0}}$, which entails $w \in \inver{R}[a_0] \cup \set{a_0}$. Since $\gf,w \md \B\neg\phi_{c_{0}} \land \bd\top$, we see that $w \not\in \set{\tup{t,k,l}:\tup{s,n,m}\rightsquigarrow\tup{t,k,l}}$ and $w \neq a'$. Thus, $w \in \set{a_i : i \in \omega}$.
    
    For (6), the right-to-left direction is again straightforward. For the other direction, suppose $\gf,w \md \psi_{A}$. Take any $u \in W$ such that $\gf,V,u \md \psi_{A}$. By (4) $u = a_i$ for some $i \in \omega$. By $\gf,V,u \md p_A \wedge \lnot\D p_A$, we see that $u$ is an $R$-maximal point in $V(p_A) \cap \set{a_i : i \in \omega}$. Since $R$ is a linear order on $V(p_A) \cap \set{a_i : i \in \omega}$, we have $V(\psi_A) = \set{a_i} = \set{w}$. As $V(\D p_A \land \lnot\D\D p_A) \cap \set{a_i : i \in \omega} = \set{a_{i+1}}$, we have $V(\psi^+_{A}) = \set{a_{i+1}}$.

    For (10), suppose $V(\pi) = \{a_k\}$ and $V(\kappa) = \{b_l\}$. Then clearly, $\gf,V,\tup{t,k,l} \md \sigma(t,\pi,\kappa)$. For the other direction, suppose $\gf,V,w \md \sigma(t,\pi,\kappa)$. Then, $\gf,V,w \md \D\phi_{c_{t}} \wedge \D\pi$, which entails $w = \tup{t',k',l'}$ for some $\tup{t',k',l'}$. By $\gf,V,w \md \D\phi_{c_{t}} \wedge \B\neg\phi_{c_{t+1}}$, we have $w \in \inver{R}[c_{t}] \setminus \inver{R}[c_{t+1}]$ and so $t'=t$. Note that $\gf,V,a_{k+1} \md \D\pi \land \B\lnot\phi_{c_{0}}$. By $\gf,V,w \md \D\pi \wedge \B\lnot(\D\pi \land \B\lnot\phi_{c_{0}})$, we see that $w \in \inver{R}[a_{k}] \setminus \inver{R}[a_{k+1}]$ and so $k'=k$. Similarly, we obtain $l=l'$. Thus, $w = \tup{t,k,l}$.
\end{proof}

Next, we encode the behavior of the Minsky machine $\M$ by formulas. With each instruction $I$ in $\M$, we associate a formula $AxI$ as follows.
\begin{itemize}
    \item $AxI\coloneqq \neg\alpha \wedge \is{5}\sigma(t,\psi_{A},\psi_B) \to \neg\alpha \wedge \is{5}\sigma(t',\psi^+_{A},\psi_B)$, if $I=t\to\tup{t',1,0}$.
    \item $AxI\coloneqq \neg\alpha\wedge\is{5}\sigma(t,\psi_{A},\psi_B)\to\neg\alpha\wedge\is{5}\sigma(t',\psi_{A},\psi^+_B)$, if $I=t\to\tup{t',0,1}$.
    \item $AxI\coloneqq [\neg\alpha\wedge\is{5}\sigma(t,\psi^+_{A},\psi_B)\to\neg\alpha\wedge\is{5}\sigma(t',\psi_{A},\psi_B)]$\newline
    $\land [\neg\alpha\wedge\is{5}\sigma(t,\phi_{a_{0}},\psi_B)\to\neg\alpha\wedge\is{5}\sigma(t'',\phi_{a_{0}},\psi_B)]$, if $I=t\to\tup{t',-1,0} (\tup{t'',0,0})$.
    \item $AxI\coloneqq [\neg\alpha\wedge\is{5}\sigma(t,\psi_{A},\psi^+_B)\to\neg\alpha\wedge\is{5}\sigma(t',\psi_{A},\psi_B)]$\newline
    $\land [\neg\alpha\wedge\is{5}\sigma(t,\psi_{A},\phi_{b_{0}})\to\neg\alpha\wedge\is{5}\sigma(t'',\psi_{A},\phi_{b_{0}})]$, if $I=t\to\tup{t',0,-1} (\tup{t'',0,0})$.
\end{itemize}
Each formula encodes the behavior of the corresponding instruction. Let $AxM\coloneqq\bigwedge_{I \in M} AxI$. This is a well-defined formula since there are only finitely many instructions in $\M$.

\begin{lemma} \label{lem:AxM-works}
    For each configuration $\tup{t,k,l}$, if $\tup{s,n,m} \rightsquigarrow \tup{t,k,l}$, then 
    \[\lnot\alpha \land \is{5}\sigma(s,\phi_{a_{n}},\phi_{b_{m}}) \to \lnot\alpha \land \is{5}\sigma(t,\phi_{a_{k}},\phi_{b_{l}}) \in \TL{K4} \oplus AxM.\]
\end{lemma}
\begin{proof}
    The proof proceeds by induction on the length of the computation of $\M$. Consider the computation of the form $\tup{s,n,m} \rightsquigarrow \tup{t,k,l} \to \tup{t',k',l'}$, where the last step is an application of $I$. As the induction hypothesis, we have 
    \[\lnot\alpha \land \is{5}\sigma(s,\phi_{a_{n}},\phi_{b_{m}}) \to \lnot\alpha \land \is{5}\sigma(t,\phi_{a_{k}},\phi_{b_{l}}) \in \TL{K4} \oplus AxM.\]
    We divide cases according to the shape of $I$. 
    
    Case (1): $I = t \to \tup{t',1,0}$. It suffices to show that 
    \[\lnot\alpha \land \is{5}\sigma(t,\phi_{a_{k}},\phi_{b_{l}}) \to \lnot\alpha \land \is{5}\sigma(t',\phi_{a_{k+1}},\phi_{b_{l}}) \in \TL{K4} \oplus AxM.\]
    This is clear by applying the substitution $[\D^k\phi_{a_{0}}/p_A, \D^l\phi_{b_{0}}/p_B]$ to $AxI$.

    Case (2): $I=t\to\tup{t',-1,0} (\tup{t'',0,0})$. Suppose $k \neq 0$. Then by applying to $AxI$ the substitution $[\D^{k-1}\phi_{a_{0}}/p_A, \D^l\phi_{b_{0}}/p_B]$, we have
    \[\lnot\alpha \land \is{5}\sigma(t,\phi_{a_{k}},\phi_{b_{l}}) \to \lnot\alpha \land \is{5}\sigma(t',\phi_{a_{k-1}},\phi_{b_{l}}) \in \TL{K4} \oplus AxM.\]
    If $k = 0$, then by applying to $AxI$ the substitution $[\D^l\phi_{b_{0}}/p_B]$, we have
    \[\lnot\alpha \land \is{5}\sigma(t,\phi_{a_{0}},\phi_{b_{l}}) \to \lnot\alpha \land \is{5}\sigma(t'',\phi_{a_{0}},\phi_{b_{l}}) \in \TL{K4} \oplus AxM.\]
    
    The cases $I=t\to\tup{t',0,1}$ and $I=t\to\tup{t',0,-1} (\tup{t'',0,0})$ follow similarly.
\end{proof}

\begin{lemma}\label{lem:FMsnm-validates-AxM}
    $\gf\md AxM$.
\end{lemma}
\begin{proof}
    Take any instruction $I \in \M$. Suppose that $I$ has the form $t\to\tup{t',1,0}$. Take any point $w$ and any valuation $V$ on $\gf$ such that $\gf,V,w\md\neg\alpha\wedge\is{5}(\sigma(t,\psi_{A},\psi_B))$. Then $\gf,V,u\md\sigma(t,\psi_{A},\psi_B)$ for some $u\in W$. Since $\gf,V,u \md \D\psi_{A} \wedge \D\psi_{B}$, by Lemma~\ref{lem:FMsnm-ch}~(6) and (8), there exists $k, l \in \omega$ such that $V(\psi_{A}) = \set{a_k}$ and $V(\psi_{B}) = \set{b_l}$. By Lemma~\ref{lem:FMsnm-ch} (10), $u=\tup{t,k,l}$ and so $\tup{s,n,m}\rightsquigarrow\tup{t,k,l}$ by the definition of $W$. Then, since $I\in \M$, we have $\tup{s,n,m}\rightsquigarrow\tup{t',k+1,l}$. By Lemma~\ref{lem:FMsnm-ch} (7), (8) and (10), $\gf,V,\tup{t',k+1,l}\md\sigma(t',\psi^+_{A},\psi_B)$. So, we have $\gf,V,w \md \neg\alpha \wedge \is{5}(\sigma(t',\psi^+_{A},\psi_B))$. Thus, $\gf\md AxI$. 
    Similarly, $\gf\md AxI$ for instructions $I\in \M$ of other forms. Hence, $\gf\md AxM$.
\end{proof}

Now we construct the reduction, extending that in \cite{Chagrov.Shehtman1995}. For each configuration $\tup{t,k,l}$, we define the logic
\begin{align*}
    L(t,k,l) \coloneqq \TL{K4} & \oplus AxM \oplus (\lnot\alpha\land\is{5}\sigma(s,\phi_{a_{n}},\phi_{b_{m}})\to\lnot\alpha\land\is{5}\sigma(t,\phi_{a_{k}},\phi_{b_{l}}))\to\alpha\\
    & \oplus (\neg\alpha\to\all{5}(\phi_{c_{t_0}}\to\phi_0)\wedge\is{5}\phi_{c_{t_0}})
\end{align*}
where
\begin{center}
    $\phi_0=(\bb(\bb(p\to\bb p)\to p)\to\bb p) \wedge \bb((\B q \wedge \neg q) \to \bd(\B^2q \wedge \D\neg q)) \wedge \bd\B^{t_0+2}\bot$
\end{center}
and $p$ and $q$ are fresh variables w.r.t. $AxM$ and $\alpha$. Note that the axioms of $L(t,k,l)$ are computable from a configuration $\tup{t,k,l}$. The subsequent lemmas show some properties of $L(t,k,l)$. 
Intuitively, the formula $\phi_0$ is designed to be such that any extension of $\TL{K4}\oplus\phi_0$ is Kripke incomplete. More precisely, we have
\begin{lemma}\label{lem:no-frame-for-phi0}
    Let $\ff=(Y,S)$ be a Kripke frame such that $\ff \md \TL{K4}$. Then, $\ff,y\not\md\phi_0$ for any $y \in Y$.
\end{lemma}
\begin{proof}
    Suppose that $\ff,y\md\phi_0$ for a contradiction. Since $\ff,y\md\bd\B^{t_0+2}\bot$, $\ff,b'\md\B^{t_0+2}\bot$ for some irreflexive $b'\in\inver{S}[y]$. Now take any irreflexive point $z\in\inver{S}[y]$. Let $V_z$ be a valuation on $\ff$ such that $V_z(q)=R[z]$. Then $\ff,z\md\B q\wedge\neg q$. By $\ff,y\md\bb((\B q\wedge\neg q)\to\bd(\B^2q\wedge\D\neg q))$, we have $\ff,z\md\bd(\B^2q\wedge\D\neg q)$ and so $\ff,V_z,z'\md\B^2q\wedge\D\neg q$ for some $z'\in\inver{S}[z]$. Then, $z'$ is again irreflexive because otherwise $\ff, V_z, z \md q$, and $z' \in \inver{S}[y]$ by the transitivity. Thus, there exists an infinite $\inver{S}$-chain $\set{z_i:i\in\omega}$ of irreflexive points in $\inver{S}[y]$. Let $V$ be a valuation on $\ff$ such that $V(p)=\set{z_{2j}:j\in\omega}$. Then we see that $\ff,V,y\not\md\bb(\bb(p\to\bb p)\to p)\to\bb p$, which contradicts $\ff,y\md\phi_0$.
\end{proof}

On the other hand, the following lemma holds:

\begin{lemma}\label{lem:at0-validates-phi0} 
    $\gf,c_{t_0}\md\phi_0$.
\end{lemma}
\begin{proof}
    First, we show that $\gf,c_{t_0}\md\bb(\bb(p\to\bb p)\to p)\to\bb p$. Take any valuation $V$ on $\gf$. Then $V(p) \in A$, so either $V(p) \cap \set{c_i : i \in \omega}$ is finite or $\set{c_i:i<\omega}\setminus V(p)$ is finite. Let $\mm=(\gf,V)$. Suppose $\mm,c_{t_0}\not\md\bb(\bb(p\to\bb p)\to p)\to\bb p$ for a contradiction. Then $\mm,c_{t_0}\md\bb(\bb(p\to\bb p)\to p)$ and $\mm, c_{t_0} \md \bd \lnot p$. Thus, $\inver{R}[c_{t_0}]\setminus V(p) \neq \emptyset$, and for every $w\in\inver{R}[c_{t_0}]\setminus V(p)$, we have $\mm,w\md\bd(p\wedge\bd\neg p)$. Since $\inver{R}[c_{t_0}]=\set{c_i:i>t_0}$, it follows that both $\set{c_i:i<\omega}\cap V(p)$ and $\set{c_i:i<\omega}\setminus V(p)$ are infinite, which is a contradiction. Thus, $\gf,c_{t_0}\md\bb(\bb(p\to\bb p)\to p)\to\bb p$.
    
    To show that $\gf,c_{t_0}\md\bb((\B q\wedge\neg q)\to\bd(\B^2q \wedge \D\neg q))$. Take any point $w\in\inver{R}[c_{t_0}]$ and any valuation $V$ in $\gf$. Note that $\inver{R}[c_{t_0}] = \set{c_i: i > t_0}$ since $\M$ contains $t_0$ many states labeled as $0, \dots, t_0-1$. Then $w = c_i$ for some $i>t_0$. Let $\mm=(\gf,V)$. Suppose $\mm,c_i \md \B q\wedge\neg q$. Then $\mm,c_{i+1}\md\B^2q\wedge\D\neg q$, which entails $\mm,c_{i}\md\bd(\B^2q\wedge\D\neg q)$. Thus, $\gf,c_i\md(\B q\wedge\neg q)\to\bd(\B^2q\wedge\D\neg q)$, and so $\gf,c_{t_0}\md\bb((\B q\wedge\neg q)\to\bd(\B^2q\wedge\D\neg q))$.
    
    Finally, $\gf,c_{t_0}\md\bd\Box^{t_0+2}\bot$ follows from $\gf,c_{t_0+1}\md\Box^{t_0+2}\bot$. Thus, we conclude that $\gf, c_{t_0} \md \phi_0$.
\end{proof}

The following two lemmas summarize how the reduction works.

\begin{lemma}\label{lem:tkl-notin-T-implies-FMsnm-validates-Ltkl}
    Let $\tup{t,k,l}$ be a configuration such that $\tup{s,n,m}\not\rightsquigarrow\tup{t,k,l}$. Then the following hold.
    \begin{enumerate}
        \item $\gf\md L(t,k,l)$.
        \item $L(t,k,l)$ is Kripke incomplete.
        \item $L(t,k,l)$ is undecidable.
    \end{enumerate}   
\end{lemma}
\begin{proof}
    For (1), we already know from Lemma~\ref{lem:FMsnm-validates-AxM} that $\gf\md AxM$. Since $\tup{t,k,l} \notin W$, by Lemma~\ref{lem:FMsnm-ch} (2), (3), and (10), we have $\gf\md\neg\sigma(t,\phi_{a_{k}},\phi_{b_{l}})$ and so $\gf\md\neg\is{5}\sigma(t,\phi_{a_{k}},\phi_{b_{l}})$. Similarly, since $\tup{s,n,m}\rightsquigarrow\tup{s,n,m}$, we have $\gf,\tup{s,n,m}\md\sigma(s,\phi_{a_{n}},\phi_{b_{m}})$. By Lemma~\ref{lem:FMsnm-basic}, $\gf\md\is{5}\sigma(s,\phi_{a_{n}},\phi_{b_{m}})$. Take any $w \in W$ and valuation $V$ in $\gf$. If $\gf, V, w \not\md \alpha$, then $\gf, V, w \md \lnot\alpha \land \is{5}\sigma(s,\phi_{a_{n}},\phi_{b_{m}}) \land \lnot\is{5}\sigma(t,\phi_{a_{k}},\phi_{b_{l}})$, i.e. $\gf,V,w \not\md \lnot\alpha\land\is{5}\sigma(s,\phi_{a_{n}},\phi_{b_{m}})\to\lnot\alpha\land\is{5}\sigma(t,\phi_{a_{k}},\phi_{b_{l}})$. Thus, $\gf\md(\lnot\alpha\land\is{5}\sigma(s,\phi_{a_{n}},\phi_{b_{m}})\to\lnot\alpha\land\is{5}\sigma(t,\phi_{a_{k}},\phi_{b_{l}}))\to\alpha$. 
    By Lemmas~\ref{lem:at0-validates-phi0} and \ref{lem:FMsnm-ch} (1), we have $\gf\md\phi_{c_{t_0}}\to\phi_0$ and so $\gf\md\all{5}(\phi_{c_{t_0}}\to\phi_0)$. By Lemmas~\ref{lem:FMsnm-basic} and \ref{lem:FMsnm-ch} (1), we have $\gf\md\is{5}\phi_{c_{t_0}}$. Thus, $\gf\md\all{5}(\phi_{c_{t_0}}\to\phi_0)\wedge\is{5}\phi_{c_{t_0}}$. Hence, $\gf\md L(t,k,l)$.

    For (2), we first recall that $\gf \not\md \alpha$ follows from the assumption. By (1), we obtain that $\gf \md L(t,k,l)$ and so $\alpha \notin L(t,k,l)$. Thus, it suffices now to show that $\alpha\in\Log(\Fr(L(t,k,l)))$. Take any Kripke frame $\ff=(Y,S)\in\Fr(L(t,k,l))$. Suppose $\ff\not\md\alpha$. Then $\ff,V,w\md\neg\alpha$ for some $w\in Y$ and valuation $V$ on $\ff$. Since $\ff\md L(t,k,l)$, we have $\ff,w\md\neg\alpha\to\all{5}(\phi_{c_{t_0}}\to\phi_0)\wedge\is{5}\phi_{c_{t_0}}$. Thus, for any valuation $V'$ that agrees with $V$ on the variables occurring in $\alpha$, we have $\ff,V',w\md\neg\alpha$ and so $\ff,V',w\md\all{5}(\phi_{c_{t_0}}\to\phi_0)\wedge\is{5}\phi_{c_{t_0}}$. Note that $\phi_0$ and $\alpha$ have no common variable and $\phi_{c_{t_0}}$ is variable-free. So, we have $\ff,w\md\all{5}(\phi_{c_{t_0}}\to\phi_0)\wedge\is{5}\phi_{c_{t_0}}$. Since $\phi_{c_{t_0}}$ is variable-free, there exists $u\in S_\sharp^5[w]$ such that $\ff,u\md\phi_{c_{t_0}}$, and so $\ff,u\md\phi_0$, which contradicts Lemma~\ref{lem:no-frame-for-phi0}. Thus, $\ff \md \alpha$, and we conclude $\alpha\in\Log(\Fr(L(t,k,l)))$.

    For (3), it suffices to show that the following holds for any configuration $\tup{t',k',l'}$,
    \begin{center}
        (\dag)  $\tup{s,n,m} \rightsquigarrow \tup{t',k',l'}$ if and only if $\neg\alpha \wedge \is{5}\sigma(s,\phi_{a_{n}},\phi_{b_{m}}) \to \neg\alpha \wedge \is{5}\sigma(t',\phi_{a_{k'}},\phi_{b_{l'}}) \in L(t,k,l)$.
    \end{center}
    This yields a reduction from the problem $\set{\tup{t',k',l'}: \tup{s,n,m} \rightsquigarrow \tup{t',k',l'}}$ to the decision problem for $L(t,k,l)$, which implies the undecidability of $L(t,k,l)$. The left-to-right direction of (\dag) follows from Lemma~\ref{lem:AxM-works}. Suppose $\tup{s,n,m} \not\rightsquigarrow \tup{t',k',l'}$. Since $\gf \not\md \alpha$, similar to the proof of (1), we see that $\gf \not\md \neg\alpha \wedge \is{5}\sigma(s,\phi_{a_{n}},\phi_{b_{m}}) \to \neg\alpha \wedge \is{5}\sigma(t',\phi_{a_{k'}},\phi_{b_{l'}})$. By (1) $\gf \md L(t,k,l)$, we have $\neg\alpha \wedge \is{5}\sigma(s,\phi_{a_{n}},\phi_{b_{m}}) \to \neg\alpha \wedge \is{5}\sigma(t',\phi_{a_{k'}},\phi_{b_{l'}}) \notin L(t,k,l)$, which concludes the other direction.
\end{proof}

\begin{lemma} \label{Lem:reduction-K4t}
    For each configuration $\tup{t,k,l}$, if $\tup{s,n,m}\rightsquigarrow\tup{t,k,l}$ then $L(t,k,l)= \TL{K4} \oplus \alpha$.
\end{lemma}
\begin{proof}
    Let $L = \TL{K4} \oplus \alpha$. It is clear that $L(t,k,l)\sub L$. Suppose that $\tup{s,n,m}\rightsquigarrow\tup{t,k,l}$. By Lemma~\ref{lem:AxM-works}, $\lnot\alpha\land\is{5}\sigma(s,\phi_{a_{n}},\phi_{b_{m}})\to\lnot\alpha\land\is{5}\sigma(t,\phi_{a_{k}},\phi_{b_{l}})\in \TL{K4}\oplus AxM$. Thus, we obtain $\alpha\in L(t,k,l)$, since $(\lnot\alpha\land\is{5}\sigma(s,\phi_{a_{n}},\phi_{b_{m}})\to\lnot\alpha\land\is{5}\sigma(t,\phi_{a_{k}},\phi_{b_{l}}))\to\alpha \in L(t,k,l)$, and thus $L(t,k,l)=L$.
\end{proof}

Now we are ready to prove the main theorem. 

\begin{proof}[Proof of Theorem~\ref{thm:reduction-K4t-P}]
    Let $P$ be a property and $\alpha \in \Lt$ be a formula such that $\gf\not\md\alpha$ and $\TL{K4}\oplus\alpha \in P$ and $P \sub \mathsf{KC} \cup \mathsf{DEC}$. The reduction $\tup{t,k,l} \mapsto L(t,k,l)$ satisfies the following:
    \begin{itemize}
        \item If $\tup{s,n,m}\rightsquigarrow\tup{t,k,l}$, then $L(t,k,l)= \TL{K4} \oplus \alpha$ by Lemma~\ref{Lem:reduction-K4t}, so $L(t,k,l) \in P$,
        \item If $\tup{s,n,m}\not\rightsquigarrow\tup{t,k,l}$, then $L(t,k,l) \notin \mathsf{KC} \cup \mathsf{DEC}$ by Lemma~\ref{lem:tkl-notin-T-implies-FMsnm-validates-Ltkl}, so $L(t,k,l) \notin P$.
    \end{itemize}
    Thus, we obtain a reduction from the set $\set{\tup{t,k,l}:\tup{s,n,m}\rightsquigarrow\tup{t,k,l}}$, which is undecidable, to the set $\set{\phi\in\Lt:\TL{K4}\oplus\phi \in P}$, which is therefore also undecidable. 
\end{proof}

As a corollary of Theorem~\ref{thm:reduction-K4t-P}, we obtain the following undecidable properties in $\NExt\TL{K4}$.

\begin{corollary} \label{Thm:undec-properties}
    The following sets are undecidable:
    \begin{enumerate}
        \item $\set{\phi\in\Lt:\TL{K4}\oplus\phi \text{ is Kripke complete}}$,
        \item $\set{\phi\in\Lt:\TL{K4}\oplus\phi \text{ is canonical}}$,
        \item $\set{\phi\in\Lt:\TL{K4}\oplus\phi \text{ is first-order definable}}$,
        \item $\set{\phi\in\Lt:\TL{K4}\oplus\phi \text{ has the FMP}}$,
        \item $\set{\phi\in\Lt:\TL{K4}\oplus\phi \text{ is locally tabular}}$,
        \item $\set{\phi\in\Lt:\TL{K4}\oplus\phi \text{ is tabular}}$,
        \item $\set{\phi\in\Lt:\TL{K4}\oplus\phi \text{ is decidable}}$,
        \item $\set{\phi\in\Lt:\TL{K4}\oplus\phi = L}$, where $L$ is an arbitrarily fixed tabular logic,
        \item $\set{\phi\in\Lt:\TL{K4}\oplus\phi \text{ is consistent}}$.
    \end{enumerate}
\end{corollary}
\begin{proof}
    All properties in (1) - (7) are contained in $\mathsf{KC} \cup \mathsf{DEC}$. Since the logic $\Log(\bullet) = \TL{K4} \oplus \B\bot \wedge \bb\bot$ is tabular, it has all properties in (1) - (7). Also, it is clear that $\gf \not\models \B\bot \wedge \bb\bot$. So, applying Theorem~\ref{thm:reduction-K4t-P} with $\alpha = \B\bot \wedge \bb\bot$, we obtain the undecidability of (1) - (7). For (8), take any tabular tense logic $L$. It follows from \cite[Theorem 3.9]{Chen.Ma2024} that there exists a formula $\alpha$ such that $L = \TL{K4} \oplus \alpha$ (this also follows from a general fact in universal algebra proved by Birkhoff \cite{Birkhoff1935}; see also \cite[Theorem 5.27]{Bergman2012}), and every extension of $L$ is again tabular. Since $\Log(\gf)$ is non-tabular, we have $\gf \not\md \alpha$. By Theorem~\ref{thm:reduction-K4t-P}, (8) is undecidable. For (9), it suffices to show the complement, namely, the inconsistency is undecidable. This follows from Theorem~\ref{thm:reduction-K4t-P} by taking $\alpha = \bot$.
\end{proof}

\begin{remark}
    Since $\TL{K4} = \TL{K} \oplus \D\D p \to \D p$, it follows from Proposition~\ref{prop:finite-extension} that the properties mentioned in Corollary~\ref{Thm:undec-properties} are also undecidable in the lattice $\NExt\TL{K}$. 
\end{remark}

\begin{remark}
    Recall that consistency is decidable in $\NExt\K$ because the lattice $\NExt\K$ has exactly two coatoms, namely, maximal consistent logics. The result that consistency is undecidable in $\NExt\TL{K4}$ aligns with fact that there are $2^{\aleph_0}$ many coatoms in the lattice $\NExt\TL{K4}$ \cite{Chen.Ma2024}.
\end{remark}

\section{Conclusion} \label{Sec 4}

The main motivation behind this paper is to understand how interactions of modalities affect the decidability of logics' properties. We have shown that most properties, as listed in Corollary~\ref{Thm:undec-properties}, are undecidable in the lattice $\NExt\TL{K4}$. These results, together with the known ones (see Table~\ref{table:UndecidableProperties}), suggest that the interactions of modalities only make the decision problem for properties harder. We would conjecture that, for a property $P$ such as Kripke completeness, the finite model property, or decidability, if $P$ is undecidable in $\NExt L$ for a unimodal logic $L$, then it is also undecidable in $\NExt L_t$. Note that this does not follow directly from the minimal tense extension map $(\cdot)_t: \NExt\mathsf{K} \to \NExt\mathsf{K}_t$, as the map may not preserve or reflect these properties, as discussed in the introduction. 

Moreover, another observation from Table~\ref{table:UndecidableProperties} is that a stronger base logic tends to make a property decidable. The undecidability results for $\NExt\TL{K4}$ indicate that $\TL{K4}$ is not strong enough in this sense. We leave it for future research to analyze the decidability of logics' properties in the lattice of extensions of other, stronger or incomparable bimodal logics. For example, the tense logic $\TL{S4}$ of pre-orders, $\TL{Grz}$ of posets and the product logic $\ML{S5}\times\ML{S5}$ would be natural choices.

\section*{Acknowledgment}

The authors would like to thank Nick Bezhanishvili for his valuable comments on the draft of this paper. The authors are also grateful for the anonymous reviewers' comments, which significantly improved the presentation of the paper. The first author is supported by Tsinghua University's Initiative for Advancing First-Class and World-Leading Disciplines in the Humanities and Social Sciences. The second author was supported by the Student Exchange Support Program (Graduate Scholarship for Degree Seeking Students) of the Japan Student Services Organization.

\bibliographystyle{eptcs}
\bibliography{Reference-aiml}

\end{document}